\documentclass{article}
\usepackage{amsmath,amsthm,amsfonts}
\usepackage{fullpage}
\usepackage{graphicx}
\usepackage{dsfont}

\newcommand{\Exp}[1]{\mathbb{E}\left[#1\right]}

\newcommand{\id}{{\mathds 1}}

\newcommand{\alow}{\alpha^{\textrm{low}}}
\newcommand{\ahigh}{\alpha^{\textrm{high}}}
\newcommand{\abuc}{\alpha^*_\textrm{BUC}}

\newcommand{\pfree}{p_\textrm{free}}

\newcommand{\ox}{\overline{x}}
\newcommand{\ds}{\mathrm{d}s}
\newcommand{\dt}{\mathrm{d}t}

\theoremstyle{definition}

\newtheorem{theorem}{Theorem}

\newtheorem{conjecture}{Conjecture}

\begin{document}

\title{The Power of Choice for Random Satisfiability}

\author{
Varsha Dani \\ University of New Mexico
\and Josep Diaz \\ Universitat Polit\`ecnica de Catalunya
\and Thomas Hayes \\ University of New Mexico
\and Cristopher Moore \\ Santa Fe Institute}

\maketitle

\begin{abstract}
We consider Achlioptas processes for $k$-SAT formulas.  We create a semi-random formula with $n$ variables and $m$ clauses, where each clause is a choice, made on-line, between two or more uniformly random clauses.  Our goal is to delay the satisfiability/unsatisfiability transition, keeping the formula satisfiable up to densities $m/n$ beyond the satisfiability threshold $\alpha_k$ for random $k$-SAT.  We show that three choices suffice to delay the transition for any $k \ge 3$, and that two choices suffice for all $3 \le k \le 25$.  We also show that two choices suffice to lower the threshold for all $k \ge 3$, making the formula unsatisfiable at a density below $\alpha_k$.  
\end{abstract}

\section{Introduction}


The Erd\H{o}s-R\'enyi model of random graphs undergoes a celebrated phase transition.  Specifically, suppose we form a random graph $G(n,m)$ with $n$ vertices and $m$ edges by choosing $m$ times uniformly from the ${n \choose 2}$ possible edges.  The average degree of this graph is $d=2m/n$.  If $d < 1$, then with high probability in the limit $n \to \infty$, $G(n,m)$ consists almost entirely of trees, and the largest component has size $O(\log n)$.  But if $d > 1$, then with high probability $G(n,m)$ has a giant connected component containing $\Theta(n)$ vertices.

In 2001, Dimitris Achlioptas posed the following question.  Suppose at each step we are presented with \emph{two} uniformly random edges.  We are allowed to choose between them, adding one of them to the graph and throwing away the other.  We play this game on-line; that is, our choice can depend on the graph up to this point, but not on future pairs of edges.  Can we delay the appearance of the giant component, ensuring that the largest component has size $o(n)$ after $m=c n$ edges for some $c > 1/2$?  

A positive answer was given by Bohman and Frieze~\cite{BohmanF01}, who showed that two choices suffice to delay the giant up to $c = 0.535$.  Achlioptas, D'Souza, and Spencer~\cite{explosive} studied a particular rule where we choose the edge that minimizes the product of the component sizes of its endpoints, which exhibits a phenomenon they call \emph{explosive percolation}.  It is also possible to speed up the appearance of the giant component~\cite{flaxman-gamarnik-sorkin,bohman-kravitz}; Spencer and Wormald~\cite{spencer-wormald} showed that it can be brought into existence at $c = 0.334$.

In analogy with $G(n,m)$, we can consider random $k$-SAT formulas $F_k(n,m)$.  Specifically, given $n$ variables $x_1, \ldots, x_n$, we create a $k$-SAT formula by choosing $m$ clauses uniformly from the $2^k {n \choose k}$ possibilities.  The \emph{satisfiability threshold conjecture} states that there is a critical density $\alpha_k = m/n$ at which $F_k(n,m)$ undergoes a phase transition from satisfiable to unsatisfiable:
\begin{conjecture}
For each $k \ge 2$, there is a constant $\alpha_k$ such that 
\[
\lim_{n \to \infty} \Pr[\mbox{$F_k(n,\alpha n)$ is satisfiable}] = 
\begin{cases} 
1 & \alpha < \alpha_k \\
0 & \alpha > \alpha _k \, .
\end{cases} 
\]
\end{conjecture}
\noindent
This conjecture has been proved only for $k=2$~\cite{ChvatalR92,Goerdt96,FernandezVega01}, where $\alpha_2 = 1$.  For the NP-complete case $k \ge 3$, there are strong arguments from statistical physics that it is true, and very precise conjectures for the value of $\alpha_k$ from calculations using the cavity method~\cite{Moore-Mertens,MezardBook}.  

There are rigorous upper and lower bounds on $\alpha_k$ assuming it exists.  That is, there are known values $\alow_k, \ahigh_k$ such that $F_k(n,\alpha n)$ is satisfiable if $\alpha < \alow_k$ and unsatisfiable if $\alpha > \ahigh_k$.  In that case, we write $\alow_k \le \alpha_k \le \ahigh_k$.  Specifically, for $k=3$ we have~\cite{DiazSat09,HajiaghayiS03,Kirousis-SAT} 
\begin{equation}
\label{eq:alpha3}
3.52 \le \alpha_3 \le 4.4898 \, , 
\end{equation}
while the cavity method gives $\alpha_3 = 4.267$.  For arbitrary $k$, the first and second moment methods give~\cite{AchlioptasP03}
\begin{equation}
\label{eq:ach-peres}
2^k \ln 2 - O(k) \le \alpha_k <  2^k \ln 2 - \frac{\ln 2}{2} \, , 
\end{equation}
while the cavity method gives 
\[
\alpha_k = 2^k \ln 2 - \frac{1+\ln 2}{2} + O(2^{-k}) \, . 
\]

Sinclair and Vilenchik~\cite{SinclairV10} asked whether Achlioptas processes can delay the satisfiability/unsatisfiability transition for $k$-SAT.  In other words, suppose at each step we are given a choice of two clauses, each of which is uniformly random.  We choose one of them and add it to the formula, and our goal is keeping the formula satisfiable up to $m = \alpha n$ clauses for some $\alpha > \alpha_k$.  They showed that two choices are enough to delay the $2$-SAT transition up to $\alpha = 1.0002$, and also that two choices can delay the $k$-SAT transition for $k = \omega(\log n)$.  Perkins~\cite{perkins} showed that for any $k$, there is a strategy with $t$ choices, for a constant $t$, that delays the $k$-SAT transition.  In fact, his analysis shows that three choices suffice for sufficiently large $k$, and that $7$ choices suffice for all $k \ge 3$.

%

We improve these results in the following ways.  First, we give a simple, nonadaptive strategy that, given a choice between three clauses, increases the $k$-SAT threshold for all $k \ge 2$.  Secondly, we give a two-choice strategy that increases the threshold for all $3 \le k \le 25$, and we conjecture that it works for all large $k$ as well.  Finally, we give a simple two-choice strategy that lowers the threshold for all $k$.

\section{Three Choices Suffice to Raise the Threshold for all $k$}

In this section and the next, we show that a constant number of choices suffice to raise the satisfiability threshold.  Our strategy is simple and nonadaptive.  Indeed, it is oblivious to the ``topology'' of the formula, which variables appear together in clauses, and is sensitive only to the signs of the literals.  Given a choice of $t$ clauses, we choose the one with the largest number of positive literals.  

To show that the resulting $k$-SAT formula is satisfiable, we convert it into an $\ell$-SAT formula in the following way: for each $k$-SAT clause $c$, we form an $\ell$-SAT clause by taking $\ell$ of the most positive literals in $c$.  If the resulting $\ell$-SAT formula is satisfiable, then so is the original $k$-SAT formula.  In Theorem~\ref{thm:3suffices}, we use $\ell=2$; in Theorems~\ref{thm:k5to25}--\ref{thm:k3}, we use $\ell=3$.

We note that Perkins~\cite{perkins} used a similar strategy, with $\ell=2$, to show that a constant number of choices suffice for any $k$.  Here we improve his results, showing that three choices suffice.

\begin{theorem}
\label{thm:3suffices}
Three choices suffice to increase the $k$-SAT threshold for any $k \ge 2$.
\end{theorem}

\begin{proof}
As described above, our strategy is simply to take the clause $c$ with the largest number of positive literals.  We then generate a $2$-SAT formula by taking two of the most positive literals from each clause.  Specifically, if $c$ has two or more positive literals, we form a $2$-SAT clause by choosing uniformly from all such pairs; if $c$ has exactly one positive literal, we take it and choose uniformly from the $k-1$ others; and if all of $c$'s literals are negative, we choose uniformly from all ${k \choose 2}$ pairs.  

If $c$ is the most-positive of $t$ uniformly random clauses, then the probabilities that the resulting $2$-SAT clause has 0, 1, or 2 positive literals are
\begin{align}
p_0 &= 2^{-kt} \nonumber \\
p_1 &= \left( 2^{-k} (k+1) \right)^{t} - p_0 \nonumber \\
p_2 &= 1-p_0-p_1 \, . 
\label{eq:biased2}
\end{align}
If there are $m=\alpha n$ clauses, this gives a biased random $2$-SAT formula with, in expectation, $\alpha p_0 n$, $\alpha p_1 n$, and $\alpha p_2 n$ clauses of these three types.  Note that the variables appearing in each clause are independent and uniformly random.  

Recall that a $2$-SAT formula on $n$ is equivalent to a directed graph on $2n$ vertices, corresponding to the literals $x_i$ and $\ox_i$ for each $1 \le i \le n$.  Each clause $(x_i \vee x_j)$ is equivalent to a pair of edges, namely the implications $\ox_i \to x_j$ and $\ox_j \to x_i$.  The formula is satisfiable if and only if no contradictory cycle exists, leading from $x_i$ to $\ox_i$ and back to $x_i$ for some $i$.  

\emph{Unit clause propagation} is the process of satisfying a unit clause, i.e.\ a clause consisting of a single literal, and generating the unit clauses implied by it and whatever 2-clauses that variable appears in.  For instance, if $(\ox_i \vee x_j)$ is one of the 2-clauses in the formula, satisfying the unit clause $(x_i)$ will generate the unit clause $(x_j)$.  In a random formula with $\alpha p_0 n$ variables, a positive unit clause $(x_i)$ will give rise, on average, to $2 \alpha p_0$ negative unit clauses $(\ox_j)$.  Similarly, a positive unit clause will give rise, on average to $p_1$ negative ones, and so on.  Unit clause propagation is thus described by a two-type branching process, with a matrix $\alpha M$ where
\begin{equation}
\label{eq:m2}
M = \begin{pmatrix}
 p_1 & 2 p_0 \\
 2 p_2 & p_1 
\end{pmatrix} \, ,
\end{equation}
where we treat the number of negative and positive unit clauses in the current generation as a column vector and multiply by $M$ on the left.  

Given an initial unit clause $u = \begin{pmatrix} 1 \\ 0 \end{pmatrix}$ or $\begin{pmatrix} 0 \\ 1 \end{pmatrix}$, the expected population generated by the entire process is 
\[
\left( \id + \alpha M + (\alpha M)^2 + \cdots \right) \cdot u \, . 
\]
If $\alpha \lambda < 1$ where $\lambda$ is the largest eigenvalue of $M$, this series converges to $( \id - \alpha M )^{-1} \cdot u$, so in expectation just $O(1)$ unit clauses are implied by the initial one.  Intuitively, this makes it very unlikely that a contradictory loop of implications exists, and therefore suggests that the $2$-SAT formula is satisfiable with high probability.  

Indeed, this was proved by Mossel and Sen~\cite{MosselSen}.  They showed that the critical density for random $2$-SAT formulas of this kind is exactly 
\[
\alpha^* = \frac{1}{\lambda} = \frac{1}{p_1 + 2 \sqrt{p_0 p_2}} \, .
\]   
For the unbiased case $p_1 = 1/2$ and $p_0=p_2=1/4$, this reproduces the $2$-SAT threshold $\alpha_2 = 1$.  Putting in our expressions~\eqref{eq:biased2} for $p_0$, $p_1$, and $p_2$ gives
\[
\alpha^* = \frac{2^{kt/2}}{2^{-kt/2} ((k+1)^t - 1) + 2 \sqrt{1-(2^{-k}(k+1))^t}} 
\] 
For large $k$, $\alpha^*$ grows as $2^{kt/2}/2$.  If we set $t=3$, then $\alpha^*$ exceeds the $k$-SAT threshold for all $k \ge 3$.  In particular, for $k=3$ we have $\alpha^* > 4.86$, which exceeds the best known upper bound on $\alpha_3$ of $4.4898$~\cite{DiazSat09}.  For $k \ge 4$, $\alpha^*$ exceeds the first moment upper bound $2^k \ln 2$.
\end{proof}

Note that we have shown not just that three choices are enough to generate satisfiable formulas above the satisfiability threshold, but that these formulas can be satisfied in polynomial time: just use the polynomial-time algorithm for 2-SAT to find a satisfying assignment.  For the case $k=2$ and $t=2$, we have also shown that two choices raise the 2-SAT threshold to $1.203$, which improves the results of~\cite{SinclairV10,perkins}.

Note also that setting $t=1$ in the proof of Theorem~\ref{thm:3suffices} shows that the threshold for random $k$-SAT without any choices grows as $\alpha_k = \Omega(2^{k/2})$.  This is far below the second moment lower bound $\Omega(2^k)$~\cite{ach-moore-ksat,AchlioptasP03}, but the proof is much simpler.

\section{Two Choices Suffice to Raise the Threshold for $3 \le k \le 25$}

In this section we show that two choices suffice for $k$ up to $25$.  We do this by analyzing simple linear-time algorithms with differential equations.  Regrettably, these equations seem too complicated to solve analytically; thus we are not able to prove that these results hold for all $k \ge 3$, though we conjecture that they do.  

We start by showing that a particularly simple algorithm works for $5 \le k \le 25$.  We then use slightly more sophisticated algorithms to raise the threshold for $k=3$ and $k=4$.

\begin{theorem}
\label{thm:k5to25}
Two choices suffice to increase the $k$-SAT threshold for all $5 \le k \le 25$.
\end{theorem}

\begin{proof}
Our strategy is the same as before: given a choice of $t$ clauses, take the one with the most positive literals.  We then form a $3$-SAT clause by choosing uniformly from among the most-positive triplets of literals.  Analogous to~\eqref{eq:biased2}, the probability that the resulting clause has $0$, $1$, $2$, or $3$ positive literals is 
\begin{align}
p_0 &= 2^{-kt} \nonumber \\
p_1 &= \left( 2^{-k} (k+1) \right)^{t} - p_0 \nonumber \\
p_2 &= \left( 2^{-k} \left( {k \choose 2} + k + 1 \right) \right)^{\!t} - p_1 \nonumber \\
p_3 &= 1-p_0-p_1-p_2 \, . 
\label{eq:biased3}
\end{align}
Now consider the following algorithm, which we call BUC for Biased Unit Clause.  At each step it sets some variable $x$ permanently, removing clauses that agree with that setting and hence are satisfied, and shortening clauses that disagree with it.
\begin{enumerate}
\item (Forced step) If there are any unit clauses, choose one uniformly and satisfy it.
\item (Free step) Otherwise, choose $x$ uniformly from all unset variables, and set $x$ true.
\end{enumerate}
This is identical to the UC algorithm for random $k$-SAT studied by Chao and Franco~\cite{chao-franco,chao-franco2} except that, on a free step, UC flips a coin to determine the truth value of $x$.  If at any point we have two contradictory unit clauses, we simply give up rather than backtracking.  Our goal is to use differential equations to show that BUC succeeds with positive probability.  The existence of a nonuniform threshold~\cite{friedgut}, which we claim applies to these biased $3$-SAT formulas as well, then implies that they are satisfiable with high probability.

After $T$ of the variables have been set, let $S_{ij}(T)$ denote the number of $i$-clauses with $j$ positive literals, for $i=2, 3$ and $0 \le j \le i$.  Initially we have $S_{3,j}(0) = \alpha p_j n$ and $S_{2,j}(0) = 0$.  Let $q_0(T)$ and $q_1(T)$ denote the probability that the variable on the $T$th step is set false or true respectively.  Then the expected change in $S_{ij}$ at each step is
\begin{align*}
\mbox{for all $0 \le j \le 3$} \, , \quad \Exp{\Delta S_{3,j}} &= - \frac{3 S_{3,j}}{n-T} + o(1) \\
\mbox{for all $0 \le j \le 2$} \, , \quad \Exp{\Delta S_{2,j}} &= \frac{(3-j) q_1 S_{3,j} + (j+1) q_0 S_{3,j+1} - 2 S_{2,j}}{n-T} + o(1) \, . 
\end{align*}
The key fact behind these equations is that, at all times throughout the algorithm's progress, the formula consisting of the remaining clauses is uniformly random once we condition on the number of clauses of each type.  In particular, the variables appearing in each clause are uniformly random among the $n-T$ unset variables, as is the variable $x$ set on a given step.  Thus each 3-clause is either satisfied or shortened with probability $3/(n-T)$; if it has $j$ positive literals and we set $x$ false, then with probability $j/(n-T)$ it becomes a 2-clause with $j-1$ positive literals; and so on.

Rescaling to real-valued variables $t = T/n$ and $s_{ij}(t) = S_{ij}(tn)/n$ in the usual way gives the differential equations
\begin{align}
\mbox{for all $0 \le j \le 3$} \, , \quad \frac{\ds_{3,j}}{\dt} &= - \frac{3 s_{3,j}}{1-t} \label{eq:diffeq-buc-s3} \\ 
\mbox{for all $0 \le j \le 2$} \, , \quad \frac{\ds_{2,j}}{\dt} &= \frac{(k-j) q_1 s_{3,j} + (j+1) q_0 s_{3,j+1} - 2 s_{2,j}}{1-t} \, , 
\label{eq:diffeq-buc-s2} 
\end{align}
with the initial conditions $s_{3,j}(0) = \alpha p_j$ and $s_{2,j}(0) = 0$.  Then classic results~\cite{Wormald95} show that, with high probability, $S_{ij}(T) = s_{ij}(T/n) n + o(n)$ for all $T$, where $s_{ij}(t)$ is the unique solution to this system of differential equations.

The caveat to this, of course, is that a contradictory pair of unit clauses does not appear.  Standard arguments show that as long as the branching process of unit clauses stays subcritical throughout the algorithm, then the probability that no contradiction occurs, and that the algorithm succeeds in satisfying all the clauses, is $\Theta(1)$.  

Analogous to~\eqref{eq:m2}, the unit clauses obey a two-type branching process between negative and positive unit clauses, where the expected number of children of each type is within $o(1)$ of the matrix
\begin{equation}
\label{eq:buc-m}
M 
= \frac{1}{1-t} \begin{pmatrix} s_{2,1} & 2 s_{2,0} \\ 2 s_{2,2} & s_{2,1}  \end{pmatrix} \, . 
\end{equation}
We can group steps together into rounds, where each round consists of a free step followed by a cascade of forced steps.  Let $\lambda$ denote the largest eigenvector of $M$.  As long as $\lambda < 1$, the branching process is subcritical, and the total expected number $b_0, b_1$ of variables set false or true respectively in a round is
\[
\begin{pmatrix} b_0 \\ b_1 \end{pmatrix} 
= \left( \id + M + M^2 + \cdots \right) \cdot \begin{pmatrix} 0 \\ 1 \end{pmatrix} 
= ( \id - M )^{-1} \cdot \begin{pmatrix} 0 \\ 1 \end{pmatrix} \, , 
\]
where we use the fact that the initial free step in each round sets a variable true.  Averaging over many steps, but not so many that $M$ changes appreciably, the probability that a variable is set false or true is 
\[
q_0 = \frac{b_0}{b_0 + b_1} \, , \; 
q_1 = \frac{b_1}{b_0 + b_1} \, .
\]
Similar analyses of multi-type branching processes in algorithms appear in~\cite{ach-moore-3col,Kalapala}.

\begin{table}
\center
$
\begin{array}{c|cccccccc}
k & 3 & 4 & 5 & 6 & 7 & 8 & 9 & 10 \\
\abuc & 4.232 & 9.491 & 24.306 & 66.811 & 190.806 & 554.106 & 1610.88 & 4637.05 \\ 
2^k \ln 2 & & & 22.181 & 44.362 & 88.723 & 177.446 & 354.891 & 709.783 
\end{array}
$
\caption{The lower bound $\abuc$ achieved by choosing the clause with the most positive literals, and running the Biased Unit Clause algorithm on the $3$-SAT formula consisting of one of the the most-positive triplets of each clause.  For $5 \le k \le 25$, $\abuc$ exceeds the first-moment upper bound on $\alpha_k$, showing that two choices are enough to raise the threshold.}
\label{tab:buc}
\end{table}

The differential equation~\eqref{eq:diffeq-buc-s3} for $s_{3,j}$ is easy to solve: namely, $s_{3,j} = \alpha p_j (1-t)^3$.  We integrate the rest of the system~\eqref{eq:diffeq-buc-s2} numerically, and use binary search to find the largest $\alpha$, up to some precision, such that $\max_t \lambda(t) < 1$.  In Table~\ref{tab:buc} we show the resulting lower bound $\abuc$ for the first few values of $k$.  For $k=3$ and $k=4$, $\abuc$ is below the conjectured values of the threshold~\cite{mertens:mezard:zecchina:03}, namely $4.267$ and $9.931$.  But for $5 \le k \le 25$, $\abuc$ exceeds the first moment upper bound $2^k \ln 2$.
\end{proof}

Asymptotically, $\abuc$ seems to grow roughly as $2.5^k$.  It is tempting to think that we can prove a lower bound on $\abuc$ sufficient to show that two choices suffice for all $k > 25$ as well, but we have not been able to do that.  

The next two theorems use slight improvements to Theorem~\ref{thm:k5to25} to raise the threshold for $k=3$ and $k=4$.

\begin{theorem}
\label{thm:k4}
Two choices suffice to increase the $4$-SAT threshold.
\end{theorem}

\begin{proof}
Given two clauses, we again take the one with more positive clauses, but now we apply the BUC algorithm directly to the resulting $4$-SAT formula.  Most of the analysis of Theorem~\ref{thm:k5to25} goes through unchanged, except that the probability that a clause has a given number of positive literals is now 
\[
p_0 = \frac{1}{256} \, , \;
p_1 = \frac{3}{32} \, , \;
p_2 = \frac{3}{8} \, , \;
p_3 = \frac{13}{32} \, , \;
p_4 = \frac{31}{256} \, .
\]
The differential equations~\eqref{eq:diffeq-buc-s2} for the density of $2$-clauses and the matrix $M$ for the branching process of unit clauses~\eqref{eq:buc-m} are unchanged.  The differential equations for $4$- and $3$-clauses are now
\begin{align}
\mbox{for all $0 \le j \le 4$} \, , \quad \frac{\ds_{4,j}}{\dt} &= - \frac{4 s_{4,j}}{1-t} \nonumber \\ 
\mbox{for all $0 \le j \le 3$} \, , \quad \frac{\ds_{3,j}}{\dt} &= \frac{(4-j) q_1 s_{4,j} + (j+1) q_0 s_{4,j+1} - 3 s_{3,j}}{1-t} \, ,
\label{eq:diffeq-buc-s4} 
\end{align}
and the initial conditions are $s_{4,j}(0) = \alpha p_j$ and $s_{3,j}(0) = s_{2,j}(0) = 0$. 

Integrating this system numerically, we find that $M$'s largest eigenvalue $\lambda$ is less than $1$ up to $\alpha = 10.709$.  This is less than the naive first moment upper bound on $\alpha_4$, but it exceeds an improved upper bound of $10.217$ based on counting locally maximal assignments~\cite{dubois-boufkhad}.
\end{proof}

Finally, we use a biased version of the Short Clause (SC) algorithm, which Chvatal and Reed used to prove a lower bound on the $3$-SAT threshold~\cite{ChvatalR92}, to show that two choices can delay the satisfiability transition in $3$-SAT.

\begin{theorem}
\label{thm:k3}
Two choices suffice to increase the $3$-SAT threshold.
\end{theorem}

\begin{proof}
Once again our strategy is to take the more positive of the two clauses.  The probability that a clause has a given number of positive literals is
\[
p_0 = \frac{1}{64} \, , \;
p_1 = \frac{15}{64} \, , \;
p_2 = \frac{33}{64} \, , \;
p_3 = \frac{15}{64} \, . 
\]
We now analyze the following algorithm, which we call Biased Short Clause (BSC).
\begin{enumerate}
\item (Forced step) If there are any unit clauses, choose one uniformly and satisfy it.
\item (Free step) Otherwise, if there are any 2-clauses, choose one uniformly.  If it has any positive literals, choose one uniformly and satisfy it.  If both its literals are negative, choose one uniformly and satisfy it.
\item (Really free step) If there are no unit clauses or 2-clauses, choose $x$ uniformly from the unset variables and choose $x$'s truth value uniformly.
\end{enumerate}
This is identical to Short Clause~\cite{ChvatalR92} except that, whenever possible, we satisfy the chosen 2-clause by setting a variable true.

During the critical phase of the algorithm, there are $\Theta(n)$ 2-clauses, so we can effectively ignore the possibility of a really free step.  Let $\pfree$ denote the probability that a given step is free.  The differential equations for 3- and 2-clauses are then
\begin{align}
\mbox{for all $0 \le j \le 3$} \, , \quad \frac{\ds_{3,j}}{\dt} &= - \frac{3 s_{3,j}}{1-t} \label{eq:diffeq-buc-s3} \\ 
\mbox{for all $0 \le j \le 2$} \, , \quad \frac{\ds_{2,j}}{\dt} &= \frac{(k-j) q_1 s_{3,j} + (j+1) q_0 s_{3,j+1} - 2 s_{2,j}}{1-t} 
- \pfree \frac{s_{2,j}}{s_{2,0}+s_{2,1}+s_{2,2}} \, , 
\label{eq:diffeq-bsc} 
\end{align}
where the additional term is due to the fact that we choose and satisfy a random 2-clause on every free step.  

As before, consider a round consisting of a free step followed by a cascade of forced steps, and let $b_0$ and $b_1$ denote the total expected number of variables set false or true during a round.  The probability that a given step is free is $1$ divided by the expected length of the round, 
\[
\pfree = \frac{1}{b_0+b_1} \, ,
\]
and the probability that a given step sets a variable false or true is $q_0 = b_0/(b_0+b_1)$ and $q_1 = b_1/(b_0+b_1)$ respectively.  The matrix $M$ describing the branching process of unit clauses is the same as in BUC.  However, the initial population of unit clauses in each round is different.  Rather than always setting a variable true, a free step sets a variable true if the chosen 2-clause has at least one positive literal, and otherwise it sets a variable false.  Thus
\[
\begin{pmatrix} b_0 \\ b_1 \end{pmatrix} 
= \frac{1}{s_{2,0}+s_{2,1}+s_{2,2}} \, ( \id - M )^{-1} \cdot \begin{pmatrix} s_{2,0} \\ s_{2,1}+s_{2,2} \end{pmatrix} \, .
\]

Integrating this system numerically, we find that $M$'s largest eigenvalue $\lambda$ stays below $1$ for all $t$ as long as $\alpha < 4.581$.  This exceeds the best known upper bound $\alpha_3 < 4.4898$, completing the proof.
\end{proof}

All these results show that two choices are enough to create a formula at a density above $\alpha_k$ that can be satisfied, with probability $\Theta(1)$, in linear time.

\section{Two Choices Suffice to Lower the Threshold, If There Is One}

We now show that two choices are enough to lower the satisfiability threshold if the threshold exists.  If there is no threshold, we can still lower it; we explain below what we mean by this tongue-in-cheek statement.

\begin{theorem}
\label{thm:lower2}
Two choices suffice to lower the threshold for $k$-SAT for any $k$, assuming that the threshold conjecture holds.
\end{theorem}

\begin{proof}
Our strategy depends on the topology of the formula, but in a very simple way.  Let $0 < a < 1$ be a constant to be determined.  We simply prefer clauses whose variables are all in the set $U = \{ x_1, x_2, \dots, x_{an} \}$ to those with one or more variables outside $U$.  

If we have $t$ choices, the probability that the chosen clause has all its variables in $U$ is
\[
q = 1-(1-a^k)^t \, , 
\]
If the subformula consisting of these clauses is unsatisfiable, then so is the entire formula.  But this subformula is uniformly random in $F_k(n',m')$ where $n' = an$ and $\Exp{m'} = qm$.  By the Chernoff bound, its density is arbitrarily close to 
\begin{equation}
\label{eq:gamma}
\alpha' = \frac{m'}{n'} = \alpha \gamma 
\quad \text{where} \quad
\gamma = \frac{1-(1-a^k)^t}{a} \, . 
\end{equation}
Thus the chosen formula is unsatisfiable w.h.p.\ if $\alpha > \alpha_k / \gamma$, lowering the threshold by a factor of $\gamma$.  

To confirm that there is an $a$ such that $\gamma > 1$, we maximize $\gamma$ as a function of $a$.  Specifically, if $t=2$ then $\gamma$ is maximized at 
\[
a = \left( \frac{2k-2}{2k-1} \right)^{1/k} \, ,
\]
where
\begin{equation}
\label{eq:gammamax}
\gamma 
= \frac{4k(k-1)}{(2k-1)^2} \left( \frac{2k-1}{2k-2} \right)^{1/k} 
\ge 1 + \frac{1}{4k^2} \, . 
\end{equation}
This completes the proof.
\end{proof}

We remark that a similar strategy shows that two choices are enough to create a giant component with $m=cn$ edges where $c = (3/8) \sqrt{3/2} = 0.459$.

What if we don't take the threshold conjecture for granted?  Theorems~\ref{thm:3suffices}--\ref{thm:k3} still ``raise the threshold'' unconditionally, in the sense that two or three choices give formulas that are w.h.p.\ satisfiable at densities where random $k$-SAT formulas are w.h.p.\ unsatisfiable.  We can give an analogous result for lowering the threshold:

\begin{theorem}
\label{thm:lower}
For any $k$, there is a constant $t$ such that $t$ choices suffice to generate formulas that are w.h.p.\ unsatisfiable at densities where random $k$-SAT formulas are w.h.p.\ satisfiable.  For sufficiently large $k$, two choices suffice.
\end{theorem}

\begin{proof}
Following the proof of Theorem~\ref{thm:lower2}, we just have to ensure that $\gamma > \gamma_k$ where $\gamma_k = \ahigh_k / \alow_k$ is the ratio between the best known upper and lower bounds on the threshold, i.e.\ the lowest and highest densities where random $k$-SAT formulas are known to be unsatisfiable or satisfiable respectively.  

Examining~\eqref{eq:gamma}, we see that for any $k$ and any $\gamma_k$ there are $a, t$ such that $\gamma > \gamma_k$.  For instance, let $a = 1/(2\gamma_k)$ and let $t$ be large enough so that $(1-a^k)^t < 1/2$.

For large $k$, from~\eqref{eq:ach-peres} we have $\gamma_k = 1+O(2^{-k} k)$, where $O$ represents a constant independent of $k$.  Since from~\eqref{eq:gammamax} we can achieve $\gamma = 1+\Theta(1/k^2)$ with two choices, there is some $k_0$ such that two choices suffice for all $k \ge k_0$.
\end{proof}

For $3$-SAT in particular, where the current value of $\gamma_k$ is $4.898 / 3.52 = 1.275$, maximizing $\gamma$ as a function of $a$ shows that $6$ choices suffice to lower the threshold unconditionally.

\section{Conclusion}

We have shown that three choices are enough to raise the satisfiability threshold in random $k$-SAT, and that two are enough to lower it, for any $k$.  We have also shown that two are enough to raise it for $k \le 25$.  We are left with several questions.
\begin{enumerate}
\item Are two choices enough to raise the threshold for any $k$?  This seems incontrovertible, but we not see how to extend our analysis of Biased Unit Clause to arbitrary $k$.
\item Sinclair and Vilenchik~\cite{SinclairV10} point out that if we are allowed to choose off-line, i.e.\ if we are given all pairs of clauses in advance, then with two choices can raise the $k$-SAT threshold exactly to the $2k$-SAT threshold, since a choice of two $k$-SAT clauses is equivalent to a $2k$-SAT clause.  Can we do nearly this well in the on-line version?  Or is there a stricter upper bound on how high we can raise the $k$-SAT threshold with two on-line choices, say $O(2^{ck})$ for some $c < 2$? 
\item Our two-choice strategy for lowering the threshold does so by a factor of $1+O(1/k^2)$.  Is there a strategy with two choices, or a constant number of choices, that lowers the threshold by a constant factor for all $k$?

\end{enumerate}

\paragraph*{Acknowledgments}  We are grateful to Stephan Mertens and Will Perkins for helpful conversations.  T.H. and C.M. are supported in part by NSF grant CCF-1219117.

\end{document}